\documentclass[a4paper,USenglish,cleveref,autoref,thm-restate]{lipics-v2021}

\usepackage[utf8]{inputenc}
\usepackage[T1]{fontenc}
\usepackage{amsmath,amsfonts,amsthm,amssymb}
\usepackage{xfrac}
\usepackage[sanserif,basic]{complexity}
\usepackage[algoruled, lined, linesnumbered]{algorithm2e}
\usepackage{setspace}
\usepackage{paralist}
\usepackage{adjustbox}
\graphicspath{{figures/}}

\usepackage{thm-restate}

\usepackage{regexpatch}

\makeatletter
\xpatchcmd\thmt@restatable{%
\csname #2\@xa\endcsname\ifx\@nx#1\@nx\else[{#1}]\fi
}{%
\ifthmt@thisistheone
\csname #2\@xa\endcsname\ifx\@nx#1\@nx\else[{#1}]\fi
\else
\csname #2\@xa\endcsname[{Restated}]
\fi}{}{}
\makeatother

\usepackage[dvipsnames]{xcolor}
\usepackage{tikz}
\usepackage{xspace}
\usepackage{footnote}

\let\oldsqrt\sqrt
\def\hksqrt{\mathpalette\DHLhksqrt}
\def\DHLhksqrt#1#2{\setbox0=\hbox{$#1\oldsqrt{#2\,}$}\dimen0=\ht0
   \advance\dimen0-0.2\ht0
   \setbox2=\hbox{\vrule height\ht0 depth -\dimen0}%
   {\box0\lower0.4pt\box2}}
\renewcommand\sqrt\hksqrt

\renewcommand{\leq}{\leqslant}
\renewcommand{\geq}{\geqslant}

\newcommand{\A}{{\cal A}}
\newcommand{\G}{{\cal G}}
\renewcommand{\H}{{\cal H}}
\newcommand{\In}{\textsc{in}}
\newcommand{\Out}{\textsc{out}}
\newcommand{\inedges}{\textsc{In-Edges}}
\newcommand{\outedges}{\textsc{Out-Edges}}
\newcommand{\nmc}{\textsc{NMC}}
\newcommand{\fmc}{\textsc{FMC}}
\renewcommand{\E}{\mathcal{E}}
\newcommand{\F}{\mathfrak{F}}

\newcommand{\minus}{\setminus}
\newcommand{\ftrs}{\textsc{FTRS}}

\newcommand{\ftbfp}{\textsc{FT-BFP}}
\newcommand{\maxflow}{\textsc{max-flow}}
\newcommand{\setcover}{\textsc{SET-COVER}}
\renewcommand{\path}{\textsc{path}}

\newtheorem{property}[theorem]{Property}
\newtheorem{problem}[theorem]{Problem}
\newtheorem*{question*}{Question}

\providecommand{\ignore}[1]{} 

\author{Shivam Bansal}%
{Department of Computer Science and Engineering, IIT Delhi, India}%
{Shivam.Bansal@ee.iitd.ac.in}%
{} % Orcid (optional)
{} % author specific funding (optional)

\author{Keerti Choudhary}%
{Department of Computer Science and Engineering, IIT Delhi, India}%
{keerti@iitd.ac.in}%
{} % Orcid (optional)
{} % author specific funding (optional)

\author{Harkirat Dhanoa}%
{Department of Computer Science and Engineering, IIT Delhi, India}%
{Harkirat.Dhanoa@ee.iitd.ac.in}%
{} % Orcid (optional)
{} % author specific funding (optional)

\author{Harsh Wardhan}%
{Department of Electrical Engineering, IIT Delhi, India}%
{Harsh.Wardhan.ee318@ee.iitd.ac.in}%
{} % Orcid (optional)
{} % author specific funding (optional)

\authorrunning{S.~Bansal, K.~Choudhary, H.~Dhanoa, H. Wardhan}
\Copyright{Shivam Bansal, Keerti Choudhary, Harkirat Dhanoa, and Harsh Wardhan}

\title{Fault-Tolerant Bounded Flow Preservers
}



\nolinenumbers %uncomment to disable line numbering

\hideLIPIcs  %uncomment to remove references to LIPIcs series (logo, DOI, ...), e.g. when preparing a pre-final version to be uploaded to arXiv or another public repository

\EventEditors{}
\EventNoEds{3}
\EventShortTitle{MFCS 2022}
\EventAcronym{MFCS}
\EventYear{2022}
\EventDate{}
\EventLocation{}
\EventLogo{}
\SeriesVolume{}
\ArticleNo{}
\begin{document}

\maketitle

\begin{abstract}

\noindent
Given a directed graph $\G = (V, E)$  with $n$ vertices, $m$ edges and a designated source vertex $s\in V$, 
we consider the question of finding a sparse subgraph $\H$ of $\G$ that preserves the flow from $s$ up to a given threshold~$\lambda$ even after failure of $k$ edges. 
We refer to such subgraphs as $(\lambda,k)$-fault-tolerant bounded-flow-preserver ($(\lambda,k)$-FT-BFP).
Formally, for any $F \subseteq E$ of at most $k$ edges and any $v\in V$, the 
$(s, v)$-max-flow in $\H \minus F$ is equal to 
$(s, v)$-max-flow in $\G \minus F$, if the latter is bounded by $\lambda$, and at least $\lambda$ otherwise.
Our contributions are summarized as follows:

\begin{enumerate}\itemsep=4pt
\item %\textsc{Upper-Bound:} 
We provide a polynomial time algorithm that given any graph $\G$ constructs a $(\lambda,k)$-FT-BFP
of $\G$ with at most 
$\lambda 2^kn$ edges. 

\item %\textsc{Lower-Bound:} 
We also prove a matching lower bound of $\Omega(\lambda 2^kn)$ on the size of $(\lambda,k)$-FT-BFP.
In particular, we show that for every $\lambda,k,n\geq 1$, there exists an $n$-vertex directed graph
whose optimal $(\lambda,k)$-\ftbfp ~contains $\Omega(\min\{2^k\lambda n,n^2\})$ edges.

\item %\textsc{Hardness of Approximation:} 
Furthermore, we show that the problem of computing approximate $(\lambda,k)$-FT-BFP
is NP-hard for any approximation ratio that is better than $O(\log(\lambda^{-1} n))$.

%\item 
%As an application of our $(\lambda,k)$-FT-BFP construction we show that 
%for any constants $\lambda,k\geq 1$, we can compute an $O(n^2)$-sized data-structure
%that given any $x,y\in V$ and set $F$ of $k$ edges, reports $(x,y)$-$\lambda$-reachability
%in near linear time.
\end{enumerate}

%prove that the problem of approximating the size of the optimal $(\lambda,k)$-FT-BFP 

\end{abstract}
\section{Introduction}
We address the problem of computing single-source fault-tolerant bounded-flow-preservers for directed graphs. The 
objective is to construct a sparse subgraph that preserves the flow value up to a parameter $\lambda$ from a given fixed source $s$, even after failure of up to $k$ edges.

%We consider the problem of constructing a sparse subgraph that preserves the flow value up to a 
%parameter  $\lambda$ from a given fixed source vertex in a directed graph, even after the failure of 
%up to $k$ edges. We refer to this problem as the 
%single-source fault-tolerant bounded-flow-preserver (FT-BFP) problem. 

The following definition provides a precise characterization of this subgraph.

%The following definition characterizes this subgraph precisely.

\begin{definition}%[$(\lambda,k)$-FT-BFP]
Let $\G=(V,E)$ be a directed graph and $s\in V$ be a designated source vertex. A $(\lambda,k)$-Fault-Tolerant Bounded-Flow-Preserver $((\lambda,k)$-\ftbfp$)$ for $\G$ is a subgraph $\H=(V,E_{\H}\subseteq E)$ of $\G$ satisfying that for every $F\subseteq E$ of at most $k$ edges, and every $t\in V$, 
$$\maxflow(s,t,\H-F)
=
\begin{cases} 
%\text{At least }\lambda,~\text{if }\maxflow(\G-F,S,t)\geq \lambda\\
%\maxflow(G-F,S,t)~\text{otherwise}.
\maxflow(s,t,\G-F)&\text{if~~}\maxflow(s,t,\G-F)\leq \lambda,\\
\text{At least }\lambda,&\text{otherwise}.\\
\end{cases}$$\vspace{-2mm}

%$$\maxflow(H-F,S,t)\geq \min\big(\lambda,\maxflow(G-F,S,t)\big)$$.\vspace{-5mm}
%$(s,v)$-max-flow in $\G- F$ is bounded by $\lambda$ then $(s,v)$-max-flow in $H-F$ and $G-F$ are identical, otherwise $(s,v)$-max-flow in $H-F$ is at least $\lambda$.
\label{definition:flow-preserver}
\end{definition}\vspace{-2mm}

For the special case of $\lambda=1$, 
the problem is referred to as $k$-Fault-Tolerant Reachability Subgraph ($k$-FTRS) in the literature.
%the problem has been referred to as $k$-FTRS in literature, 
Here the goal
is to preserve reachability from $s$ after $k$ edge failures. Baswana~et~al.~\cite{BaswanaCR:16} showed
that there exists a $k$-FTRS with at most $2^k n$ edges.
Lokshtanov et al.~\cite{lokshtanov2019brief} presented an algorithm for computing a 
$(\lambda,k)$-\ftbfp ~for directed graphs. 
Their algorithm runs in time ${O}(4^{k+\lambda}(k+\lambda)^2(m+n)\cdot m)$, 
and each vertex of the \ftbfp~has in-degree at most $4^{k+\lambda}(k+\lambda)$. 
They also showed that a $(k+\lambda-1)$-FTRS of a graph $\G$ also serves as it's $(\lambda,k)$-\ftbfp. 
Using this result 
in conjunction with the algorithm from~\cite{BaswanaCR:16}, they obtain an alternate construction of a 
$(k,\lambda)$-\ftbfp ~with at most $2^{k+\lambda} n$ edges.
However, this bound is quadratic in $n$ for any $\lambda$ larger than $\log n$.

%\newpage
%\noindent
We consider the problem of obtaining a tight bound on $(\lambda,k)$-\ftbfp.
Specifically, we aim to answer the following question:
\\ \\[-2mm]
{\em %{\bf \em Question.}
Given a directed graph $\G = (V, E)$ with a source $s$, and a flow threshold $\lambda \geq \log n$, 
	can we construct a {\em sparse} $(\lambda,k)$-\ftbfp ~$\H = (V,E_H\subset E)$?
	If so, can we present graphs for which the construction turns out to be tight? }
%
%\begin{quote}\em 
%Given a directed graph $\G = (V, E)$ with a source $s$, and a flow threshold $\lambda \geq \log_2 n$, 
%	can we construct a {\em sparse} $(\lambda,k)$-\ftbfp ~$\H = (V,E_H\subset E)$?
%	If so, can we present graphs for which the construction turns out to be tight?
%\end{quote}
\vspace{2mm}

In this paper, we affirmatively answer the question above.
We provide construction for \ftbfp ~ that has a linear dependence on $\lambda$.
In particular, we prove the following.
\vspace{0.5mm}

\begin{theorem}
There exists an algorithm that for any directed graph~$\G$ on $n$ vertices
and $m$ edges, and 
any integers $\lambda,k\geq 1$, computes in
$O(\lambda2^k mn)$ time a $(\lambda,k)$-\ftbfp ~for $\G$ with at most $\lambda2^k n$ edges.
%Moreover, the in-degree of each vertex in this $(\lambda,k)$-\ftbfp ~is bounded by $2^k\lambda$.
\end{theorem}
\vspace{0.5mm}

Furthermore, we show that the extremal bound of $\lambda2^k n$ in the above construction is tight by 
presenting the following lower bound.
\vspace{0.5mm}

\begin{theorem}\label{theorem:main-lo-bound}
%For any $k,\lambda,n\geq 1$ satisfying $\lambda 2^k\leq n$, there exists an $n$-vertex directed graph
%whose optimal $(\lambda,k)$-\ftbfp ~contains $\Omega(2^k\lambda n)$ edges.
For every $\lambda,k,n\geq 1$, there exists a construction of an $n$-vertex directed graph %$\G$
whose optimal $(\lambda,k)$-\ftbfp ~contains $\Omega(\min\{2^k\lambda n,n^2\})$ edges.
\end{theorem}
\vspace{0.5mm}

%A natural question that exists next is whether for graphs whose $(\lambda,k)$-\ftbfp
%is know to be sparse can we obtain near optimal constructions.

We next consider the problem of approximating  $(\lambda,k)$-\ftbfp ~structures.
We show that unless $P=NP$, there is no polynomial-time algorithm to obtain a
$\log (\lambda^{-1}n))$ approximation to optimal $(\lambda,k)$-\ftbfp.
%We present our hardness result by providing a reduction from the set-cover problem.
%Our hardness result is as follows.

\vspace{0.5mm}

\begin{theorem}\label{HardnessThm1}
For any $\lambda,k,n\geq 1$ satisfying $k=\Omega(\log (\lambda^{-1}n))$, 
the problem of computing an $O(\log (\lambda^{-1}n))$ approximate $(\lambda,k)$-\ftbfp~ for $n$ vertex digraphs is NP-hard.
\end{theorem}

\vspace{0.5mm}

We finally %extend our \ftbfp ~construction to directed graphs with non-unit capacities. 
%We also 
present
application of our \ftbfp ~construction in computing all-pairs
fault-tolerant $\lambda$-reachability oracle.
%In particular, we show that for any constants $\lambda,k\geq 1$, we can compute an $O(n^2)$-sized data-structure
%that given any $x,y\in V$ and set $F$ of $k$ edges, reports $\lambda$-reachability from $x$ to $y$
%in near linear time.

%\begin{theorem}
%Given any directed graph $\G=(V,E)$ on $n$ vertices, an integer $\lambda\geq 1$, and a constant $k\geq 1$, we can preprocess $\G$ in polynomial time to build an $O(\lambda n^2)$ size data structure that, given any query vertex-pair $(x,y)$ and any set $F$ of $k$ edges, reports the $(x,y)$\hspace{0.1mm} $\lambda$-reachability in $\G\setminus F$ in 
%$O(\lambda n^{1+o(1)})$ time.
%\end{theorem}

\begin{theorem}
Given any directed graph $\G=(V,E)$ on $n$ vertices and any positive constants $\lambda,k\geq 1$, we can preprocess $\G$ in polynomial time to build an $O(n^2)$ size data structure that, given any query vertex-pair $(x,y)$ and any set $F$ of $k$ edges, reports the $(x,y)$\hspace{0.1mm} $\lambda$-reachability in $\G\setminus F$ in 
$O(n^{1+o(1)})$ time.
\end{theorem}

%To the best of our knowledge no non-trivial constructions for all-pairs $k$-fault-tolerant $\lambda$-reachability
%were known earlier. for $\lambda\geq \Omega(\log n)$.

%\begin{theorem}
%For any constants $\lambda,k\geq 1$, we can compute an $O(n^2)$-sized data-structure
%that given any $x,y\in V$ and set $F$ of $k$ edges, reports $(x,y)$-$\lambda$-reachability
%in near linear time.
%
%$(s,t)$-flow oracles, that have
%subquadratic space and query time.

\subsection{Existing Works}

For undirected graphs, there exists a tight construction for $(\lambda,k)$-\ftbfp~with ${O}((k+\lambda)\cdot n)$ edges that 
directly follows from $\alpha$-edge connectivity certificate constructions provided by Nagamochi and Ibaraki~\cite{NagamochiI:92}.

A closely related problem to that of graph preservers is  fault-tolerant reachability oracles.
For dual failures, the work of \cite{Choudhary16} obtained an $O(n)$ size single source reachability oracle with constant query time for directed graphs.
Brand and Saranurak~\cite{BrandS19}, showed construction of 
an $\widetilde O(n^2)$ sized 
 $k$-fault-tolerant 
all-pairs reachability oracle that has $O(k^\omega)$ query time.

Recently,
Baswana et al.~\cite{BaswanaBP22} considered the problem of 
sensitivity oracle for reporting max-flow value for a single source-destination pair. 
They presented an $O(n^2)$ size data-structure that
after failure of any two edges, reports the max-flow value of the surviving graph in constant
time. 
%To the best of our knowledge no efficient constructions are known for more than two failures.

For the problem of computing the value of all-pairs max-flow up to $\lambda$ 
%(a.k.a. all-pairs $\lambda$-reachability)
in the static setting, Abboud et at.~\cite{AbboudGIKPTUW19} 
obtained two deterministic algorithms that work for DAGs: a combinatorial algorithm which runs in ${O}(2^{{O}(\lambda^2)}\cdot mn)$ time, and another algorithm that can be faster on dense graphs which runs in ${O}((\lambda\log n)4^{\lambda+o(\lambda)}\cdot n^\omega)$ time.

Some other graph theoretic problems studied in the fault-tolerant model
 include computing distance preservers~\cite{DTCR08, PP13, Parter15},
 depth-first-search tree~\cite{BCCK19}, spanners~\cite{CLPR09, DK11},
 approximate single source distance preservers~\cite{BK13, PP14, BGLP16},
 approximate distance oracles~\cite{DP09, CLPR10}, compact routing
 schemes~\cite{CLPR10, Chechik13}.

% Another simple approach for finding $k$-fault-tolerant connectivity-preservers for undirected graphs can be found here \cite{kFTcp} which is based upon the repeated inclusion of spanning forests in the subgraph.

\section{Preliminaries}
\label{section:prelim}

Given a digraph $\G=(V,E)$ on $n=|V|$ vertices and $m=|E|$ edges, %with a designated source vertex $s$,
we first define some notations used throughout the paper.

\begin{itemize}
%\item $E(f)$:~ Edges of the graph $G$ carrying a non-zero flow for a given flow $f$.
%\item $G_f$:~ The residual graph of $G$ with respect to flow $f$.
%\item $E(P)$:~ Edges lying on a path $P$.
%\item $E(A)$:~ Edges of the graph $G$ whose both endpoints lie in set $A$, where $A\subseteq V$.
\item $\In(v,\G)$:~ The set of in-neighbours of $v$ in $\G$.
\item $\Out(v,\G)$:~ The set of out-neighbours of $v$ in $\G$.
\item $\inedges(v,\G)$:~ The set of all incoming edges of $v$ in $\G$.
\item $\outedges(v,\G)$:~ The set of all outgoing edges of $v$ in $\G$.
\item $\Out(A,\G)$:~ The set of all those vertices in $V\backslash A$ having an incoming
edge from some vertex of $A$ in $\G$, where $A\subseteq V(\G)$.
\item $\G(A)$:~ The subgraph of $\G$ induced by the vertices lying in a subset $A$ of $V$.
\item $\G+(u,v)$:~ The graph obtained by adding an edge $(u,v)$ to graph $\G$.
\item $\G\minus F$:~ The graph obtained by deleting the edges lying in a set $F$ from graph $\G$.
\item $\maxflow(S,t,\G)$:~ The value of the maximum flow in graph $\G$ from a source set $S$ to a destination vertex $t$. When the set $S$ comprises of a single vertex, say $s$, we 
represent it simply by $\maxflow(s,t,\G)$.
\item $\textsc{path}[a,b,T]$:~ The path from node $a$ to $b$ in a tree $T$.
\item $P[a,b]$:~ The subpath of path $P$ lying between vertices $a$ and $b$, where $a$ precedes $b$ on $P$.
\item $P\circ Q$~:~ The path formed by concatenating paths $P$ and $Q$ in $\G$.
Here it is assumed that the last edge (or vertex) of $P$ is the same as the first edge (or vertex) of $Q$.
\end{itemize}

We next define the concept of farthest min-cut that was introduced by Ford and Fulkerson in their pioneering work on flows and cuts~\cite{FF62}.
%
%\begin{definition}
Let $S$ be a source set, and $t$ be a destination vertex. Any $(S,t)$-cut $C$
is a partition of the vertex set into two sets: $A(C)$ and $B(C)$, where $S\subseteq A(C)$ and $t\in B(C)$. 
An $(S,t)$-min-cut $C^*$ is said to be the 
\emph{farthest min-cut} if $A(C^*)\supsetneq A(C)$ for every $(S,t)$-min-cut
$C$ other than $C^*$. We denote the cut $C^*$ by  $\fmc(S,t,\G)$.
%\label{definition:FMC}
%\end{definition}
%
Similar to farthest-min-cut, we can define the nearest min-cut. %as follows.
%\begin{definition}
%Let $S$ be a source set and $t$ be a destination vertex. 
An $(S,t)$-min-cut $C^*$ is said to be the 
\emph{nearest min-cut} if $A(C^*)\subsetneq A(C)$ for every $(S,t)$-min-cut
$C$ other than $C^*$. We denote the cut $C^*$ by  $\nmc(S,t,\G)$.
%\label{definition:NMC}
%\end{definition}

%We next 

Below we state a property of nearest and farthest $(s,t)$-min-cuts~\cite{FF62}.

\begin{property}
Let $s$ be a source vertex, $t$ be a destination vertex, and $f$ be an $s$ to $t$ max-flow in graph $\G$.
Let $\G_f$ denote the residual graph corresponding to flow $f$.
Further let $X$ be the set of vertices reachable from $s$ in $\G_f$, and $Y$ be the set of vertices having
a path to $t$ in $\G_f$.
Then $\nmc(s,t,\G)=(X,V\setminus X)$ and $\fmc(s,t,\G)=(V\setminus Y,Y)$.
\end{property}

%We now state the property of the farthest min-cut (for proof see~\cite{BaswanaCR:16}).
%
%
%\begin{lemma}
%Let $S$ be a source set, $t$ be a destination vertex, $C$ be $\fmc(G,S,t)$, and 
%$(A,B)$ be the partition of $V$ corresponding to cut $C$. Let $(s,w)\in (S\times B)$ 
%be any arbitrary edge, and $G'=G+(s,w)$ be a new graph. Then,
%$\maxflow(G',S,t)$ $= 1 + \maxflow(G,S,t)$.
%%, and $C'=C\cup\{(s,w)\}$ is a $(S,t)$-min-cut for $G'$.
%\label{lemma:FMC-property1}
%\end{lemma}

%\section{Hardness Results}
\section{Hardness of logarithmic approximation}
%\subsection{Hardness results for computing approximate FT-BFP}

%\begin{theorem}
%For any $n$ vertex graph with a source $s$, 
%computing a $O(\log n)$ approximation of optimal $\kftrs$ for $k=\Omega(\log n)$ is NP-hard.
%\end{theorem}

We prove in this section the following hardness result for approximating optimal $\ftbfp$.	

\begin{theorem}\label{HardnessThm1}
For any $\lambda,k,n\geq 1$ satisfying $k=\Omega(\log (\lambda^{-1}n))$, 
the problem of computing an $O(\log (\lambda^{-1}n))$ approximate $(\lambda,k)$-\ftbfp~ for $n$ vertex digraphs is NP-hard.
\end{theorem}

We prove the above theorem by showing a reduction from the \setcover~ problem to the optimal \ftbfp.  

\begin{problem}[\cite{10.1007/978-3-642-32512-0_24}, Definition 1] 
\label{HardnessDef1}
The input to \setcover~ consists of base set $U$, $|U| = n$ and subsets $S_1,...,S_m \subseteq U$, $\cup_{j=1}^{m} S_j = U$, $m \leq poly(n)$. The goal is to find as few sets $S_{i_1},...,S_{i_k}$ as possible that cover $U$, that is, $\cup_{j = 1}^{k}S_{i_j} = U$
\end{problem}

\begin{lemma}[\cite{10.1007/978-3-642-32512-0_24}, Theorem 2] 
\label{HardnessLem1}
For every $0 < \alpha < 1$ (exact) SAT on inputs of size $n$ can be reduced in polynomial time to approximating \setcover~ to within $(1 - \alpha) \ln{N}$ on inputs of size $N = n^{O(1/\alpha)}$. 
\end{lemma}

From Lemma \ref{HardnessLem1}, we can also deduce that it is NP-Complete to approximate \setcover ~ up to a multiplicative factor of $c_1\log{\max(n,m)}$ for some $c_1 > 0$ as $m \leq poly(n)$.\\

\noindent
\textbf{Transformation}~
Given a \setcover ~ instance $\langle U,\F \rangle$, we will construct a $(\lambda,k)$-\ftbfp ~
instance $\langle \G,s \rangle$. The transformation is as follows (also see Figure~\ref{fig:hardness}).

\begin{enumerate}
   
\item  Round up the number for elements in $U$ to nearest power of $2$ (let this be $2^{u}$) 
by adding $2^{u} - |U|$ new elements to $U$ and all these new elements to every set in $\F$.

\item Initialize $\G$ to be the graph with $N+1$ vertices, namely, $s,v_1,\ldots,v_N$
where $N=4\lambda(m+n)$.

\item Next construct the following subgraph $\G_i$, for each $i\in[1,\lambda]$.
\begin{enumerate}
\item
Construct a complete binary tree $B_i$ rooted at a vertex $r_i$ of height $u$ and $2^u$ leaf nodes.
The leaf nodes of $B_i$ will correspond to elements in the universe $U$.
From each leaf node $x_i$ in $B_i$, 
add out-edges to two new vertices, namely, $\ell(x_i)$ and $r(x_i)$.

\item For each set $W\in \F$, add a vertex $y_{i,W}$ to graph $\G_i$. 
%We denote the resultant by $Y_i$ that comprises of $|\F|$ vertices.
Let $Y_i$ denote the resulting set which consists of $|\F|$ vertices.
For each $x\in U$ and $W\in \F$, add an edge from $\ell(x_i)$ to 
$y_{i,W}$ if and only if $x\in W$.

\item Add a set $Z_i$ of $u+1$ additional vertices.
For each leaf $x_i$ in $B_i$, add an edge from $r(x_i)$ to each vertex in the set $ Z_i$.
\end{enumerate}

\item Finally, we add an edge from $s$ to the roots $r_1,\ldots,r_\lambda$.
Also for each $i\in[1,\lambda]$, we add an edge from each vertex in $Y_i\cup Z_i$ 
to each of the vertices $v_1,\ldots,v_N$.
%\\
%and each $W\in \F$, we
%add an edge from $y_{i,W}$ to 
%%leaf node $x_i$ in $B_i$, we add an edge from $\ell(x_i)$ to 
%vertices $v_1,\ldots,v_N$.
%Also for each $i\in[1,\lambda]$,
%we add an edge from each $z\in Z_i$ to vertices $v_1,\ldots,v_N$.
\end{enumerate}

We set $k=u + 1$ for this $(\lambda,k)$-\ftbfp ~instance.
\vspace{3mm}

\begin{figure}
\centering
\includegraphics[width=0.96\textwidth]{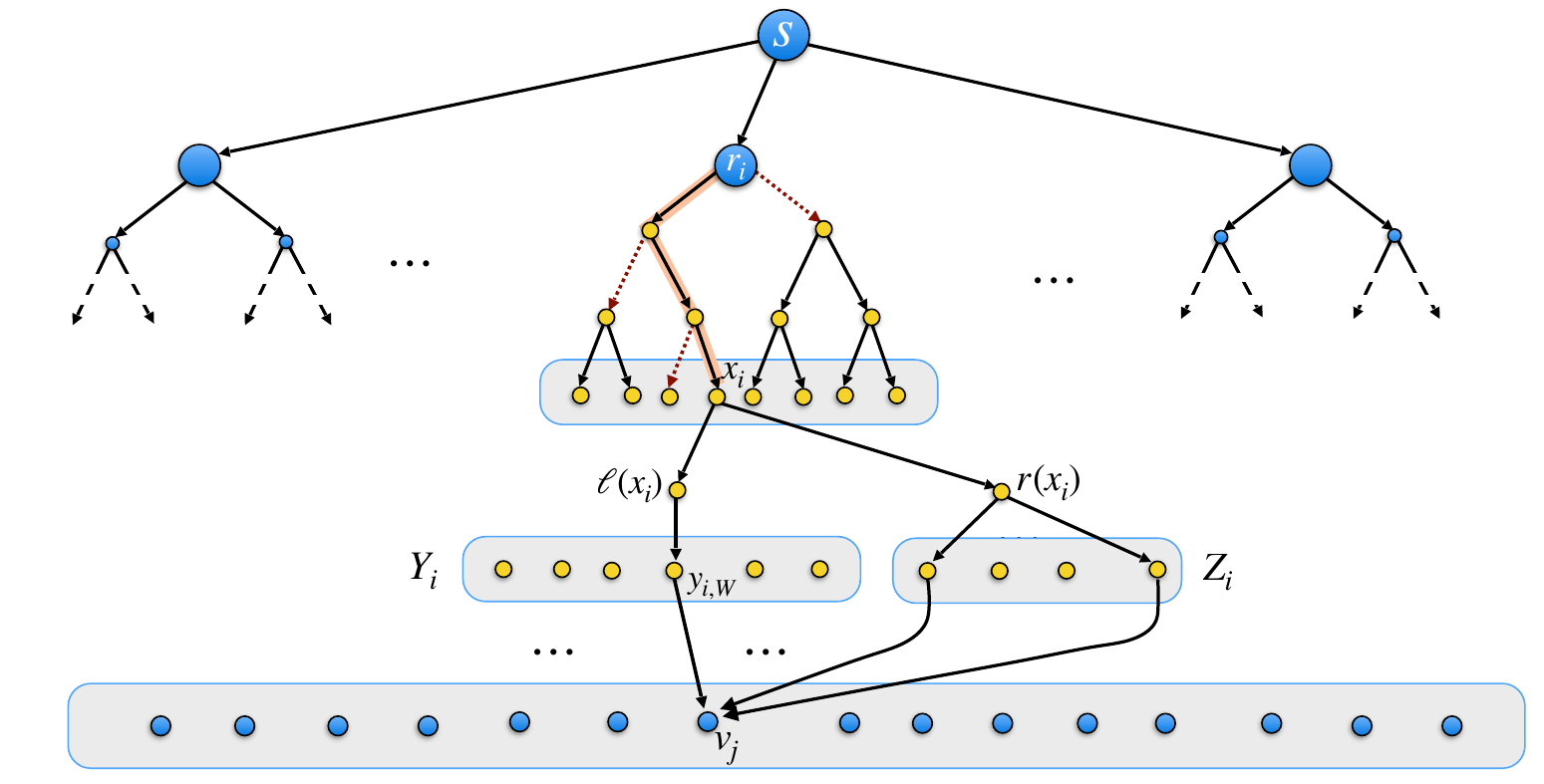}
\caption{Depiction of a $(\lambda,k)$-\ftbfp ~instance obtained from a \setcover ~instance $\langle U,\F \rangle$.}
\label{fig:hardness}
\end{figure}

%\begin{lemma} \label{HardnessLem2}
%For any $(\lambda,k)$-\ftbfp ~$\H$ of $\G$, $\forall$~ elements $x \in U$, $i \in [1,\lambda]$ there exists a set $W \in \F$ such that $x \in W$ and ${y_j}$ to $t$ edge exists in $\H$. 
%\end{lemma}
%
%\begin{proof}
%We prove this by defining set $F$ of $k$ edge failures. Let P be path from $s'_{x}$ to $e^3_{x,i}$, Let $F_1$ be the set of all those edges $(u, v)$ $\in$ binary tree rooted at $s'_x$ such that $u$ $\in$ $P$, $u \neq e^3_{x,i}$ and $v$ is the child of $u$ not lying on $P$. $F = F_1 \cup (e^3_{x,i},e^1_{x,i})$.
%     
%Since each $E_i$ is part of atleast one set $S_j$ hence there exists a $s'_x$ to $t$ path namely $P\circ e^2_{x,i}\circ s_{x,j}\circ t$ in $\G\setminus F$. Since all faults are present in $T_x$ all other substructures $T_y$ such that $y \neq x$ will also have this path from $s'_y$ to $t$. Hence $\maxflow(s,t,G\setminus F)=\lambda$. Thus $\maxflow(s,t,\H\setminus F)=\lambda$ must also hold. Since each substructure only allows 1 unit of flow through them, There must exist a path from $s'_x$ to $t$ in $\H$ for flow to be $\lambda$ in $\H\setminus F$. This path can only pass through $e^2_{x,i}$  and thus there is path from $e^2_{x,i}$ to $t$ in $\H$, hence there should exist an $s'_{x,j}$ such that $e^2_{x,i}\circ s'_{x,j}\circ t$ exists in $\H$, which implies our lemma.
%\end{proof}

\begin{lemma} \label{HardnessLem3} 
Any $(\lambda,k)$-\ftbfp ~$\H$ of the graph instance $\langle \G,s \rangle$, can be used to construct a solution of the \setcover ~instance of size at most $\lambda^{-1}(\min_{j=1}^N|\In(v_j,\H)|)$.
\end{lemma}

\begin{proof}
Consider a vertex $v_j$ in $\H$ that minimizes $|\In(v_j,\H)|$.
%Let $i\in[1,\lambda]$ be the index for which the cardinality of the set 
%$E(\G)\cap (Y_i\cup Z_i\times\{v_j\})$ is least.
%
Consider the following candidate solutions
$$S_i = \{ W\in\F  ~|~  (y_{i,W},v_j) \in E(\H) \}.$$ 
Out of the $\lambda$ sets, namely $S_1,\ldots,S_\lambda$, let $S_{i_0}$ be the set with least cardinality. 
%Since \[\sum_{S \in X} |S| \leq   |\In(t,\H)| \] 
The cardinality of $S_{i_0}$ is at most $|\In(v_j,\H)|/\lambda$ as minimum value is upper-bounded by the average value. 

Now in order to prove that $S_{i_0}$
is a valid solution, consider an element $x\in U$.
Let $P$ be the unique path from $r_{i_0}$ to leaf node $x_{i_0}$ in $B_{i_0}$, 
and let $F_1$ be the set of all those edges $(u,v)\in B_{i_0}$
such that $u\in P$ and $v$ is the child of $u$ not lying on $P$. 
Observe that $x_{i_0}$ is the unique leaf in $B_{i_0}$ that is reachable from $s$ in $\H\setminus F_1$.
Let $F_2$ be a singleton set comprising of the edge $(x_{i_0},r(x_{i_0}))$.
Consider the set $F=F_1\cup F_2$ of size $k$.
Since $\maxflow(s,v_j,\G\minus F)=\lambda$,
there must exists a path, say $Q$, from $s$ to $v_j$ in $\H\setminus F$ passing through $r_{i_0}$.
Such a path $Q$ must pass through $\ell(x_{i_0})$ as well as a vertex in $Y_{i_0}$, say $y_{i_0,W}$. 
%If this vertex is $y_{i_0,W}$, then 
This implies that
the edge $(y_{i_0,W},v_j)$ lies in $\H$, 
and so by definition of $S_{i_0}$, the set $W$ lies~in~$S_{i_0}$.
Moreover $W$ contains the element $x$ as $(\ell(x_{i_0}),y_{i_0,W})$ is an edge in $\G$.
This proves that element $x\in U$ is covered by $S_{i_0}$, and thus
%Now from Lemma \ref{HardnessLem2} we know $\forall$~ elements $E_i \in U$, $x \in [1,\lambda]$ there exists a set $S_j \in \F$ such that $E_i \in S_j$ and $s_{x,j}$ to $t$ edge exists in $\H$. 
%This~proves that~
$S_{i_0}$ is a valid solution to $\langle U,\F \rangle$.  
\end{proof}

\begin{lemma} \label{HardnessLem4}
Any solution $S$ of the \setcover~ instance $\langle U,\F \rangle$, can be used to construct a solution 
$\H$ of 
$(\lambda,k)$-\ftbfp ~instance %$\langle \G,s \rangle$ 
satisfying $|\In(v_j,\H)| = \lambda(|S| + k)$,
for each $j\in[1,N]$.
\end{lemma}

\begin{proof}
Let $S$ be a solution of the \setcover~ instance $\langle U,\F \rangle$.
Consider the sets
$$A_i = \{y_{i,W}~|~ W \in S\} \cup Z_i,\text{~for~}i\leq\lambda, \text{~~~and~~~}  A = \bigcup_{i=1}^{\lambda} A_i.$$

We will show that $$\H=\G\setminus \cup_{j=1}^N\inedges(v_j) + \cup_{j=1}^N(A\times v_j).$$ 

is a $(\lambda,k)$-\ftbfp ~of $\G$.\\ 

Let us assume, to the contrary, that $\H$ is not a $(\lambda,k)$-\ftbfp ~of $\G$.
Then there must exist an edge set $F$ of size at most $k$ and an index $j\in[1,N]$ satisfying
$\maxflow(s,v_j,\G\setminus F)$ is greater than $\maxflow(s,v_j,\H\setminus F)$.
Observe that each path from $s$ to $v_j$ must pass through a vertex $r_i$, for some 
$i\in[1,\lambda]$, and each $r_i$ only allows a unit flow to pass through~it.

Since $\maxflow(s,v_j,\G\setminus F) > \maxflow(s,v_j,\H\setminus F)$, there must exist an 
index $i\in[1,\lambda]$ satisfying that there exists a path from $s$ to $v_j$ in $\G \setminus F$ 
{\em passing through $r_i$},
but no such corresponding path exists in $\H \setminus F$. 
%This is because 
%Let $$L_0=\{x_i\in B_i~|~x_i\text{ is a leaf in }B_i\text{ reachable from }s\text{ in }\G\minus F\}$$

Let $R=\{x^0_i,x^1_i,\ldots,x^\alpha_i\}$ be the set of 
leaf nodes in tree $B_i$ {\em reachable} from $s$ in $\G\minus F$.
There exist at least $\min(k+1,|R|)$ vertex-disjoint paths 
from $R$ to $v_j$ in $\H$, namely,
\begin{itemize}
\item $(\{x^0_i,\ell(x^0_i),y_{i,W},v_j)$,
where $W\in\F$ is the set in $S$ that contains the element $x^{0}\in U$.

\item $(\{x^c_i,r(x^c_i),z^c_i,v_j)$, for $c=1$ to $\min(k,|R|-1)$.
\end{itemize}

Thus even after $k$ faults atleast one path from $r_i$ to $v_j$ will exist in $\H \setminus F$. 
This contradicts the assumption that there is no $s$ to $v_j$ path 
in $\G\setminus F$ passing through $r_i$.
Hence, $\maxflow(s,v_j,\G\setminus F)$ must be identical to $\maxflow(s,v_j,\H\setminus F)$.
%But we chose the fault $F$ such a path would not exist in $\H \setminus F$. 
%Thus no such fault $F$ exists and subsequently $\H$ is $(\lambda,k)$-\ftbfp for $\G$.
\end{proof}

%The Theorem \ref{HardnessThm1}'s proof directly follows 
%from Lemma \ref{HardnessLem1}, Lemma \ref{HardnessLem3}, Lemma \ref{HardnessLem4},
%and the fact that for every $n\geq 1$ there are hard instances of \setcover ~problem $(U,\F)$
%satisfying $|U|=n$
%whose optimal solution has size much much larger than $\log|U|$.

The proof of Theorem \ref{HardnessThm1} now directly follows from
Lemma~\ref{HardnessLem1}, Lemma~\ref{HardnessLem3}, and Lemma~\ref{HardnessLem4}, 
along with the fact that for every integer $n\geq 1$, there exist hard instances of the 
\setcover ~problem $(U,\F)$ satisfying $|U| = n$, where the size of the optimal solution is
 significantly larger than $\log |U|$.

\section{Upper bound of $\lambda 2^kn$ Edges}

In this section we will provide construction of a sparse $(\lambda,k)$-\ftbfp.

\subsection{Locality Property for Flow Preservers}

\begin{lemma}
Let $\G=(V,E)$ be a graph with a source $s\in V$, $\lambda\geq 1$ be an integer, and $v$~be~a vertex in $V$. Let $\alpha=\min\big(\lambda,\maxflow(s,v,\G)\big)$.
Let $\E_v$ be the set of in-edges of $v$ corresponding to any arbitrary set of $\alpha$-edge-disjoint paths from $s$ to $v$ in $\G$. Further, let $\H$ be a subgraph of $\G$ obtained by 
restricting the in-edges of $v$ to those present in $\E_v$. 
%removing the in-edges of $v$ that do not lie in $\E_v$. 
Then, for any vertex $t\in V$,
$$\maxflow(s,t,\H)~\geq~ \min\big(\lambda,\maxflow(s,t,\G)\big).$$
\label{lemma:flow-property}
\end{lemma}
\vspace{-6mm}

\begin{proof}
        By construction of $H$, $\alpha = \maxflow(s,v,H)$.
        Denote $\beta = \maxflow(s,t,H)$.
        Let $(A,B)$ be an $(s,t)$-min-cut in $H$.
           If $v \in A$ then, by construction of $H$,
           the $(s,t)$-cut $(A,B)$ has value $\beta$ also in $G$,
           so $\beta \geq \maxflow(s,t,G)$ and we are done.
           Assume next $v \in B$.
           Then $(A,B)$ is an $(s,v)$-cut of value $\beta$ in $H$.
           By construction of $H$, $\alpha = \maxflow(s,v,H)$, so $\beta \geq
       \alpha$.    If $\alpha = \lambda$ we are done, so assume $\alpha =
       \maxflow(s,v,G)$.
      
           We now show that $\beta \geq \maxflow(s,t,G)$, which ends the
       proof.
           Suppose not, and let $f$ be an $(s,t)$-max-flow in $H$.
           Then the residual graph $G_f$ must have an augmenting path $P$,
           containing some edges present in $G$ but not in $H$.
           Such edges are all incoming to $v$.
           Thus, $P = P[s,w] \circ (w,v) \circ P[w,t]$ where $(w,v) \in E(G)
       \setminus E(H)$,
           and $P[s,w], P[v,t]$ are present in the residual graph $H_f$.
           Adding $P$ to $f$ gives an $(s,t)$-flow of in $H + (w,v)$,
          implying that
          \begin{enumerate}
            \item $\maxflow(s,t, H+(w,v)) = \beta + 1$
            \item $(w,v) \in A \times B$
            \item $(A,B)$ is an $(s,t)$-min-cut in $H + (w,v)$
          \end{enumerate}
      
           Let $\{ Q_i \circ e_i \circ Q'_i \}_{i=1}^{\alpha}$ be $\alpha$
       edge-disjoint $s$-to-$v$ paths in $H$,
           where the edge $e_i$ of each such path is its last edge crossing the $(s,v)$-cut $(A, B)$,
           so $V(Q'_i) \subseteq B$.
           Such exist as $\alpha = \maxflow(s,v,H)$.
           Let $e_{\alpha+1}, \dots, e_{\beta}$ be the other edges crossing
       $(A,B)$ in $H$.
           Let $e_0 = (w,v)$, crossing $(A,B)$ by (ii).
           Let $\{ P_j \circ e_j \circ P'_j\}_{j=0}^{\beta}$ be $\beta+1$
       edge-disjoint $s$-to-$t$ paths in $H+(w,v)$,
           each crossing the cut $(A, B)$ exactly once, at $e_j$, so $V(P_j)
       \subseteq A$.
           Such exist by (i) and (iii).
           Then, $\{P_0 \circ e_0\} \cup \{P_i \circ e_i \circ Q'_i \}_{
       i=1}^{\alpha}$
          are $\alpha +1$ edge-disjoint $s$-to-$v$ paths in $G$,     contradicting $\alpha = \maxflow(s,v,G)$.
       \end{proof}

In the next lemma we show that in order to compute a sparse $(\lambda,k)$-\ftbfp ~it suffices to focus on a single destination node.

\begin{lemma}[Locality Lemma for Flow Preservers]
Let $\A$ be an algorithm that given any graph $\G$ and any vertex $v\in V(\G)$, 
computes a $(\lambda,k)$-\ftbfp ~of $\G$ with at most $c_{\lambda,k}$ in-edges to $v$.
Then using $\A$, one can construct for any $n$ vertex digraph a $(\lambda,k)$-\ftbfp ~with at most $c_{\lambda,k}\cdot n$ edges.
\label{lemma:locality-lemma}
\end{lemma}

\begin{proof}
Consider a graph $\G$ with $n$ vertices, namely, $v_1,\ldots,v_n$.
We will provide a construction of  $(\lambda,k)$-\ftbfp ~of $\G$ using black-box access to algorithm~$\A$. 
We compute a sequence of graphs 
$\G_0,\G_1,\ldots,\G_n$ as follows:
\begin{enumerate}
\item Initialize $\G_0=\G$.
\item For $i\geq 1$, compute $\G_i$ in two steps:
\begin{enumerate}
\item First use $\A$ to compute a $(\lambda,k)$-\ftbfp ~of $\G_{i-1}$ in which the in-degree of $v_i$ is bounded by $c_{\lambda,k}$, let this graph be~$\H_{i-1}$.
\item Obtain $\G_i$ from $\G_{i-1}$ by restricting the incoming edges of $v_i$ to those present in~$\H_{i-1}$.
\end{enumerate}
\end{enumerate}

It is easy to verify that the in-degree of each vertex in $\G_n$ is at most $c_{\lambda,k}$. 

To show that $\G_n$ is a $(\lambda,k)$-\ftbfp ~of $\G$, it suffices to show that $\G_i$ is a $(\lambda,k)$-\ftbfp ~of $\G_{i-1}$, for each $i\geq 1$.

Let us fix an index $i$ in the range $[1,n]$.
Consider a set $F$ of at most $k$ edges in $\G_{i-1}$, and
let $$\alpha=\min\big(\lambda,\maxflow(s,v_i,\G_{i-1}\minus F)\big).$$
By construction, $\H_{i-1}$ is a $(\lambda,k)$-\ftbfp ~of $\G_{i-1}$, so there exists at least $\alpha$ edge-disjoint paths from $s$ to $v_i$ in the graph $\H_{i-1}\minus F$.
Let $\E_i$ be the set of in-edges of $v_i$ corresponding to these~$\alpha$ edge-disjoint paths. 
Observe that the edges in $\E_i$ lie in graph $\G_{i}\minus F$.
Moreover, graphs $\G_i\minus F$ and $\G_{i-1}\minus F$ differ only at in-edges of $v_i$.
Therefore by~Lemma~\ref{lemma:flow-property} it follows that for any vertex $t\in V(G)$, 
$\maxflow(s,t,\G_i\minus F)\geq \min\big(\lambda,\maxflow(s,t,\G_{i-1} \minus F)\big)$. 
This proves that $\G_i$ is a $(\lambda,k)$-\ftbfp ~of $\G_{i-1}$.
\end{proof}

\subsection{Construction of an Improved FTRS}

We present here an improved bound on the in-degree of a node $t$ in $k$-\ftrs ~when the node $t$ satisfies that $(s,t)$-max-flow in $\G$ is larger than one.
In particular, we prove the following theorem.

\begin{theorem}
Let $\G$ be an $n$ vertex, $m$ edges directed graph with a designated source node~$s$.
Let $t$ be a vertex satisfying $\maxflow(s,t,\G)= f$, for some positive integer $f$.
Then for every $k\geq 1$, we can compute in $O(2^kfm)$ time a $(k+f-1)$-\ftrs ~for $\G$ in which the in-degree of node $t$ is at most $2^{k} f$.
\label{theorem:flow-FTRS}
\end{theorem}

Let us focus on a single destination node $t$.
We first show that it suffices to provide construction of $(k+f-1)$-\ftrs ~for a
graph in which out-degree of each vertex other than $s$ is bounded by $2$.
In order to prove this we will transform the graph $\G=(V,E)$ into another graph $\H=(V_H,E_H)$ satisfying
that (i) the value of $(s,t)$-max-flow in graphs $\G$ and $\H$ is identical;
(ii) the out-degree of every vertex in $\H$ other than $s$ is bounded by two.
The steps to transform $\G$ into graph $\H$ are as follows:
\begin{enumerate}
\item Initialize $\H$ to be the graph $\G$. 
\item Split each edge $e=(x,y)\in E$ by inserting two new vertices $\ell_{x,y}$
and $r_{x,y}$ between the endpoints $x$ and $y$, so that edge $(x,y)$ is translated into the path
$(x,\ell_{x,y},r_{x,y},y)$.
\item For every node $y\in V\minus\{s,t\}$ if 
$x_1,\ldots,x_p$ are in-neighbours of $y$ in $\G$ and
$z_1,\ldots,z_q$ are out-neighbours of $y$ in $\G$, then we replace vertex $y$ (in current $\H$) by $p$ binary trees as follows.
First we remove node $y$ from $\H$. Next for each $x_i\in \In(y,\G)$
insert a binary tree $B_{x_i,y}$ to $\H$ (along with new internal nodes and edges) whose root is $r_{x_i,y}$ 
and leaves are $\ell_{y,z_1},\ldots,\ell_{y,z_q}$.
\end{enumerate}

Notice that $\H$ has $O(mn)$ edges and vertices. 
Indeed for every vertex $v$ (other than $s$ and $t$) in $\G$, $|\In(v,\G)|$ binary trees have been added to $\H$, 
each of size $O(|\Out(v,G)|)$. So the number of edges and vertices in the transformed graph is
$O(\sum_{v\in V} |\In(v,G)|\cdot |\Out(v,\G)|))=O(mn)$. Also, observe that 
the out-degree of each vertex in $\H$ other than $s$ bounded by two.

\begin{lemma}
%The $(s,t)$-max-flow in graphs $\G$ and $\H$ is identical.
$\maxflow(s,t,\G)=\maxflow(s,t,\H)$
\label{lemma:identical-max-flow}
\end{lemma}
\begin{proof}
We will show that each $s$ to $t$ path in $\G$ now corresponds to a unique $s$ to $t$ path in $\H$. 
Suppose there exists a path $(s=u_0,u_1,u_2,\ldots,u_k=t)$ in $\G$. Then we will have an equivalent path in $\H$ as
\begin{align*}
(s,\ell_{u_0,u_1},r_{u_0,u_1})\circ 
& \path(r_{u_0,u_1},\ell_{u_1,u_2},B_{u_0,u_1})\circ (\ell_{u_1,u_2},r_{u_1,u_2})\circ 
\ell \dots \circ \\
&\path(r_{u_{k-2},u_{k-1}},\ell_{u_{k-1},u_k},B_{u_{k-1},u_k})\circ \ell_{u_{k-1},u_k},r_{u_{k-1},u_k})\circ 
(r_{u_{k-1},u_k},t)
\end{align*}
where $\path(r,\ell,B)$ denotes the path from $r$ to $\ell$ using edges in binary tree $B$. Therefore, the $(s,t)$-max-flow values in graphs $\G$ and $\H$ are identical.
\end{proof}

We will now justify the significance of our transformation by providing a way to construct a $(k+f-1)$-\ftrs ~of $\G$ if we know a $(k+f-1)$-\ftrs ~for $\H$ such that the in-degree of $t$ in both the 
\ftrs s is identical.

\begin{lemma}
A $(k+f-1)$-\ftrs ~for $\G$ can be constructed by knowing a $(k+f-1)$-\ftrs ~of $\H$, that preserves the in-degree of node $t$.
\label{lemma:transformation_justification}
\end{lemma}

\begin{proof}
Let $\H^*$ be a $(k+f-1)$-\ftrs ~of $\H$. We want to construct $\G^*$, a $(k+f-1)$-\ftrs ~for $\G$ satisfying the condition 
that in-degree of $t$ in graphs $\G^*$ and $\H^*$ is identical. 

The construction of $\G^*$ is as follows:
For each in-neighbour $w$ of the vertex $t$ in $\G$, include edge $(w,t)$ in $\G^*$ if and only if edge $(r_{w,t},t)$ is present in $\H^*$.
Thus, the in-degree of $t$ in graphs $\G^*$ and $H^*$ is identical.
For vertices $v$ other than $t$, we include all in-neighbours of $v$ in $\G^*$. 

We will now prove that $\G^*$ is a $(k+f-1)$-\ftrs ~of $\G$.
Consider any set $F$ of at most $k$ failed edges in $\G$.
Define a set $F_0$ of failed edges in $\H$ by including edge $(\ell_{u,v},r_{u,v})$ in $F_0$ for every $(u,v)\in F$.
From the path correspondence above and the fact that
$\H^*$ is a $(k+f-1)$-\ftrs ~of $\H$,
 it is evident that
% there is a path from $s$ to $t$ in $\G^*\minus F$ if and only if there is a path from $s$ to $t$ in 
% $\H^*\minus F'$. Furthermore, 
for any $r\leq \lambda$, there are $r$-edge-disjoint paths from $s$ to $t$ in $\G^*\minus F$ if and only if there are $r$-edge-disjoint paths from $s$ to $t$ in $\H^*\minus F_0$. %This is because two paths can pass through a vertex $v$ in $\G^*\minus F$ iff there exist two corresponding paths which pass through the binary trees which replaced $v$ in $\H^*\minus F'$.
Therefore, $\G^*$ is a $(k+f-1)$-\ftrs ~of~$\G$.
\end{proof}

It was shown in \cite{BaswanaCR:16} that if out-degree of $s$ is one,
and out-degree of all other vertices is bounded by two, then Algorithm~\ref{algo:kFTRS} computes a $k$-\ftrs ~for
$\G$ in which in-degree of $t$ is at most $2^k$.
We will prove in the next lemma that if $\maxflow(s,t,\G)=f$, and out-degree of every vertex other than $s$ is bounded by two,
then Algorithm~\ref{algo:kFTRS} in fact computes a $(k+f-1)$-\ftrs ~for
$\G$ in which the in-degree of $t$ is at most $2^k f$.

\begin{lemma}
Let $\G$ be a directed graph satisfying that the out-degree of every vertex other than the designated source $s$ is bounded by $2$, and $k\geq 1$ be an integer parameter.
Let $t\in V(\G)$ satisfy $\maxflow(s,t,\G)=f$, for some positive integer $f$.
Then Algorithm~\ref{algo:kFTRS} computes a $(k+f-1)$-\ftrs ~for $\G$ in which the in-degree of node $t$
is at most $2^{k} f$.
\end{lemma}

\begin{proof}
Consider the following algorithm from~\cite{BaswanaCR:16} for computing $k$-\ftrs ~that 
bounds in-degree of an input node $t$.

\begin{algorithm}[!ht]
\BlankLine
$S_1 \leftarrow \{s\}$\;
\For{$i=1$ to $k$}
{
    $C_i\leftarrow \fmc(S_i,t,\G)$\;
    $(A_i,B_i)\leftarrow $ Partition($C_i$)\;
% $(A_i,B_i)\leftarrow $ Partition of $\G$ induced by $C_i$\;
    $S_{i+1}\leftarrow (A_i \cup \Out(A_i,\G) )\minus \{t\}$\;
}
$f_0 \leftarrow$ max-flow from $S_{k+1}$ to $t$\;
${\cal E}(t)\leftarrow$ Incoming edges of $t$ present in $E(f_0)$\;
Return $\G^* = (\G\minus \inedges(t,\G)\big) + {\cal E}(t)$\;
\caption{Algorithm for computing $k$-$\ftrs$}
\label{algo:kFTRS}
\end{algorithm}

We will now show $\G^*$ is a $(k+f-1)$-\ftrs ~of $\G$. Let $F$ be any set of $k+f-1$ failed edges. If there exists a path $R$ from $s$ to $t$ in $\G \minus F$ then we shall prove the existence of a path $\hat{R}$ from $s$ to $t$ in $\G^*\minus F$. Observe that $R$ must pass through each $(s,t)$-cut $C_i$, for each $i \in [1,k]$, through an edge, say $(u_i,v_i)$. If $v_i=t$ then $(u_i,v_i) \in {\cal E}(t)$ and thus $R$ is intact in the graph $\G^*$. Now we need to prove for the case when the edge $(u_i,v_i)\notin {\cal E}(t)$.
%$(u_i,v_i) \in C_i\minus E^t_i$ for every $i \in [1,k]$. 

To prove that a path $\hat{R}$ exists in $\G^*$, we will construct a sequence of auxiliary graphs as done in \cite{BaswanaCR:16}, say $\H_i$’s, for each $i \in [1,k+1]$, as follows: 

\begin{center}
    $\H_1=\G$,~~~~ $\H_i=\G+ (s,v_1) +...+ (s,v_{i-1})$ , $i \in [2,k+1]$. \\
\end{center}

From the induction proof of Lemma 18 of \cite{BaswanaCR:16}, we get $\maxflow(s,t,\H_{i+1})=1+\maxflow(s,t,\H_{i})$ and since $\maxflow(s,t,H_{1})=\maxflow(s,t,\G)=f$, we get that $\maxflow(s,t,\H_{k+1})=k+f$. 
Let $\H^*= (\H_{k+1}\minus \inedges(t)) +{\cal E}(t)$ i.e. the incoming edges of $t$ are restricted in $\H_{k+1}$ to those present in the set ${\cal E}(t)$. In Lemma 19 of \cite{BaswanaCR:16} it is shown that $\maxflow(s,t,\H^{*})=\maxflow(s,t,\H_{k+1})=k+f$. Since the flow in $\H^*$ is greater than $|F|$ or the number of faults, we can directly use the Lemma 20 of \cite{BaswanaCR:16} to see that there exists a path $\hat{R}$ in $\G^* \minus F$.

The bound on the number of edges also follows from \cite{BaswanaCR:16}. Lemma 21 of \cite{BaswanaCR:16} states that $|C_{i+1}| \leq 2|C_i|$ where $C_{k+1}=FMC(S_{k+1},t,\G)$. Since $|C_1|=f$, we get the bound on  ${\cal E}(t)=C_{k+1}$ as $2^kf$. Note that the proof of Lemma 21 of \cite{BaswanaCR:16} assumes that every vertex has out-degree bounded by two but it can be shown that the Lemma will hold true even when the out-degree of all vertices except the source vertex is bounded by two by using the fact that in the proof of Lemma 21, $\Out(A_i)$ will never contain the source vertex for any $i$. 
\end{proof}

\subsection{Computing sparse $(\lambda,k)$-\ftbfp}

In this subsection, we will show how to construct a $(\lambda,k)$-\ftbfp ~of $\G$ from 
a $(k+f-1)$-\ftrs ~of $\G$. We will start by introducing a lemma from \cite{lokshtanov2019brief}, 
followed by additional lemmas that will help us to obtain a tight construction for \ftbfp.

%We will now demonstrate in this section how a $(\lambda,k)$-\ftbfp ~of $\G$ can be constructed from 
%$(k+f-1)$-\ftrs ~of $\G$. 
%We will begin by stating a Lemma from \cite{lokshtanov2019brief} and then present some additional lemmas that will be used to obtain a tight construction for \ftbfp.

%
%We shall first state a lemma from \cite{lokshtanov2019brief} and 
%improve on it to obtain a tight bound for \ftbfp ~construction.
%
%We shall first state a lemma from \cite{lokshtanov2019brief}  and present some other important lemmas that will be used to obtain a tight construction for \ftbfp.
%

\begin{lemma}[\cite{lokshtanov2019brief}]
Let $\G$ be a directed graph with a designated source node~$s$, and let
 $\H$ be a $(k+\lambda-1)$-\ftrs ~of $\G$. Then, $\H$ is also a $(\lambda,k)$-\ftbfp ~of $\G$. 
\label{lemma:older-reduction}
\end{lemma}

%We strengthen the above lemma as follows.
%The above lemma can be  strengthend as follows.

%We strengthen the above lemma by demonstrating how to construct a
%We strengthen the above lemma by providing construction of a 
%$(\lambda,k)$-\ftbfp  ~from 
%a $(\min\{f,\lambda\}+k-1)$-\ftrs, where $f$ is max-flow from $s$ to a node $t$ in the graph.

To strengthen the above lemma, we present a method for constructing a 
$(\lambda,k)$-\ftbfp ~from a $(\min\{f,\lambda\}+k-1)$-\ftrs, where 
$f$ represents the maximum flow from the source node $s$ to a destination node $t$ in the graph.

\begin{lemma}
Let $\G$ be a directed graph with a designated source node~$s$, and 
let $t$ be a vertex satisfying $\maxflow(s,t,\G)= f$, for some positive integer $f$.
Then a $(\min\{f,\lambda\}+k-1)$-\ftrs ~of $\G$ that differs from $\G$ 
only at in-edges of $t$ is a $(\lambda,k)$-\ftbfp ~for~$\G$.
\label{lemma:improved-reduction}
\end{lemma}
\begin{proof}
Let $\H$ be a $(\min\{f,\lambda\}+k-1)$-\ftrs ~of $\G$ that deviates from $\G$ 
only at in-edges of~$t$. 
It follows from Lemma~\ref{lemma:older-reduction} that the subgraph $\H$ is a 
$(\min\{f,\lambda\},k)$-\ftbfp ~for $\G$.

The claim trivially holds true if $f\geq \lambda$, so let us 
consider the scenario $f<\lambda$.
Consider a set $F$ of at most $k$ edge failures in $\G$, and 
let $p$ be $\maxflow(s,t,\G\minus F)$.
Since $p\leq f< \lambda$ and $\H$~is a $(f,k)$-\ftbfp, the max-flow from $s$ to $t$ in $\H\minus F$
must be exactly $p$. %as $p\leq \min\{f,\lambda\}$.

Since $\G$ and $\H$ only differs at in-edges of $t$,
it follows from Lemma~\ref{lemma:flow-property} that for each $v\in V(\G)$, 
$\maxflow(s,v,\H\minus F)\geq \min(\lambda,\maxflow(s,v,\G\minus F)).$
This proves that $\H$ is a $(\lambda,k)$-\ftbfp ~for~$\G$.
\end{proof}

We	now provide construction of a 
%We now show how to construct a 
$(\lambda,k)$-\ftbfp ~	that bounds 
the in-degree of a single destination node $t$.

\begin{lemma}
Let $\G$ be an $n$ vertex, $m$ edges directed graph with a designated source node~$s$,
and $t$ be any arbitrary vertex in $\G$.
Then for any $\lambda,k\geq 1$, we can compute in
$O(2^k\lambda m)$ time a $(\lambda,k)$-\ftbfp ~for $\G$ in which the in-degree of $t$ is bounded
above by $2^k \lambda$.
\label{lemma:ft-bfp-single-node}
\end{lemma}

%\begin{theorem}
%There exists an algorithm that for every $n,m,\lambda,k\geq 1$
%and every $n$ vertex, $m$ edges directed graph $\G$ with a designated source node~$s$,
%computes in $O(2^k\lambda mn)$ time a $(\lambda,k)$-\ftbfp ~for $\G$ in which the in-degree of 
%each vertex is bounded by $2^{k} \lambda$.
%\end{theorem}

\begin{proof}
%Due to Lemma~\ref{lemma:locality-lemma}, 
%it suffices to show that we can compute in $O(2^k\lambda m)$ time 
%a $(\lambda,k)$-\ftbfp ~for $\G$, say $\H$, in which in-degree of any given destination
%node is bounded by $2^{k} \lambda$.
%Consider a destination node $t\in V(\G)$, and let $f=\maxflow(s,t,\G)$.
%\\
Let $f=\maxflow(s,t,\G)$.
We present a construction of a $(\lambda,k)$-\ftbfp, say $\H$, by considering the following two cases.\\

\noindent
Case 1. $\maxflow(s,t,\G) \geq \lambda+k$: \\ 
Let us start by taking a look at the scenario $f\geq\lambda+k$. %To construct $\H$
In this case we can choose any $\lambda+k$ incoming edges of $t$ which carry a flow of $\lambda+k$ from $s$ to $t$ and discard all other incoming edges of $t$ to construct $\H$. The resulting graph $\H$ will be a $(\lambda,k)$-\ftbfp ~of $\G$ due to Lemma~\ref{lemma:improved-reduction}, and the in-degree of $t$ in $\H$ will be
$\lambda+k\leq 2^k\lambda$. 
\\

\noindent
Case 2. $\maxflow(s,t,\G) < \lambda+k$: \\ 
We next consider the case $f<\lambda+k$. In this case we use Theorem~\ref{theorem:flow-FTRS}
to compute a $(\min\{f,\lambda\}+k-1)$-\ftrs ~of $\G$, say $\H_0$, such that
the in-degree of $t$ in $\H_0$
is at most $2^k\min\{f,\lambda\}$. %which is bounded by $2^k\lambda$.
We obtain the graph $\H$ from $\G$ by limiting the incoming edges of $t$ to those present in $\H_0$.
The resulting graph $\H$ will be a $(\lambda,k)$-\ftbfp ~of~$\G$ due to 
Lemma~\ref{lemma:improved-reduction}.
\end{proof}

%On combining together Lemma~\ref{lemma:locality-lemma} and Lemma~\ref{lemma:ft-bfp-single-node}, we get the following result.

We conclude with the following theorem that directly follows by
combining together Lemma~\ref{lemma:locality-lemma} and Lemma~\ref{lemma:ft-bfp-single-node}.

\begin{theorem}
Let $\G$ be an $n$ vertex, $m$ edges directed graph with a designated source node~$s$.
Then for any $\lambda,k\geq 1$, we can compute in
$O(2^k\lambda mn)$ time a $(\lambda,k)$-\ftbfp ~for $\G$ in which the in-degree of every vertex is bounded
above by $2^k \lambda$.
\end{theorem}

%\begin{theorem}
%There exists an algorithm that for every $n,m,\lambda,k\geq 1$
%and every $n$ vertex, $m$ edges directed graph $\G$ with a designated source node~$s$,
%computes in $O(2^k\lambda mn)$ time a $(\lambda,k)$-\ftbfp ~for $\G$ in which the in-degree of 
%each vertex is bounded by $2^{k} \lambda$.
%\end{theorem}
\section{Matching Lower Bound}

We shall now show that for each $\lambda,k$, $n$ ($n\geq 3\lambda  2^k$), there exists a directed graph 
$\G$ with $O(n)$
vertices whose $(\lambda,k)$-\ftbfp~must have $\Omega(2^k\lambda n)$ edges. 

The construction of graph $\G$ is as follows. 
Let $B_1,\ldots,B_\lambda$ be vertex-disjoint complete binary trees of height $k$ rooted at vertices $r_1,\ldots,r_k$, and let $s$ be a new vertex have an edge to each of the $r_i$'s. Let $X$ denote the set of leaf nodes of these $\lambda$ trees, and let $Y$ be another set containing $n-\sum_{i=1}^\lambda|V(B_i)|-1~(\geq n/3)$ vertices.
Then the graph $G$ is obtained by adding an edge from each $x\in X$ to each $y\in Y$.
In other words, $V(G)=\{s\}\cup V(B_1)\cup\cdots \cup V(B_\lambda)\cup Y$ and 
$E(G)=\{(s,r_i)~|~1\leq i\leq \lambda\}\cup E(B_1)\cup\cdots \cup(B_\lambda)\cup(X\times Y)$.

We prove in the following lemma that any $(\lambda,k)$-\ftbfp~of the above
constructed graph contains at least  $\Omega(2^k\lambda n)$ edges.

\begin{lemma} Any $(\lambda,k)$-\ftbfp~of $\G$ must contain $\Omega(2^k\lambda n)$ edges.
\end{lemma}

\begin{proof}
It is easy to see that the out-edges of $s$, and the
edges of each of the binary tree $B_i$'s must be present in a $(\lambda,k)$-\ftbfp ~of $\G$. Thus, let us consider an edge $(x,y)\in X\times Y$,
%
%Note that $(s,y)$-max-flow $=\lambda$ for any $y\in Y$. So every $(\lambda,k)$-\ftbfp~has to contain the edge which connects $s$ to the root of a binary tree, otherwise $(s,y)$-max-flow would not remain preserved. Furthermore, every tree edge $(u,v)$ should also be present in a $(\lambda,k)$-\ftbfp~because it is the only edge which can carry flow to $v$.\\
%
%Now consider an edge $(x,y)\in X\times Y$ 
where $x$ is the leaf node of some binary tree $B_i$. 

Let $P$ be the unique path from $r_i$ to $x$ in $B_i$, and let $F$ be the set of all those edges $(u,v)\in B_i$
such that $u\in P$ and $v$ is the child of $u$ not lying on $P$. 
%We can precisely select a set $F$ of $k$ edges (one at each level of the tree) so that after these fail, 
On failure of set $F$, there remains a unique path from $s$ to $y$ that passes through edge $(s,r_i)$. 
Moreover, $\maxflow(s,y,\G\minus F)=\lambda$.
So, any subgraph $\H$ of $\G$ not containing $(x,y)$ edge would 
not be a $(\lambda,k)$-\ftbfp ~as
on failure set $F$, $\H$ would not preserve $(s,y)$-max-flow.

Hence, any $(\lambda,k)$-\ftbfp~of $\G$ contains at least $|X\times Y|=2^k\lambda|Y|\geq2^k\lambda n/3$ edges.
\end{proof}

\begin{figure}
\centering
\includegraphics[width=0.95\textwidth]{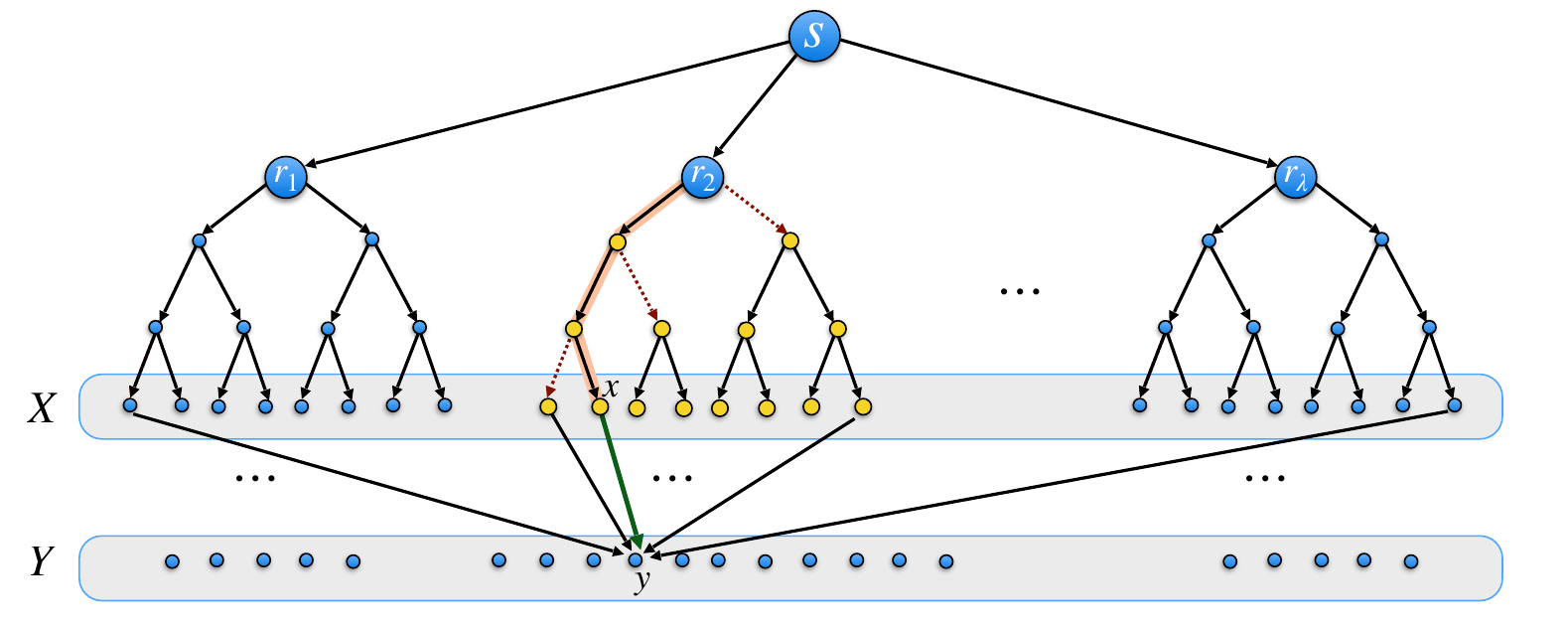}
\caption{Depiction of lower bound on the size of $(\lambda,k)$-\ftbfp ~when $k=3$.}
\end{figure}

\section{Applications}
\label{section:applications}

In this section we present applications of $\ftbfp$ structure.

\subsection{Fault-tolerant All-Pairs $\lambda$-reachability oracle}

Georgiadis et al. \cite{GeorgiadisGIPU17} showed that for any $n$ vertex directed graph $\G=(V,E)$ we can compute 2-reachability information for all pairs of vertices in $O(n^\omega \log n)$ time, where $\omega$ is the matrix multiplication exponent.
Abboud et at.~\cite{AbboudGIKPTUW19} extended this result to all-pairs $\lambda$-reachability by presenting an algorithm that takes ${O}((\lambda\log n)4^{\lambda+o(\lambda)}\cdot n^\omega)$ time. One of the interesting open questions is if for any constants 
$\lambda,k\geq 1$, we can compute an oracle that given any query vertex-pair $x,y\in V$ and any set $F$ of $k$ edge failures, reports $(x,y)$-$\lambda$-reachability in $\G\setminus F$ efficiently.
%in truly {\em sub-quadratic time}.

For any vertex $x\in V$, let $\H_x$ denote a $(\lambda,k)$-FT-BFP of $\G$ with $x$ as the source.
Our data structure simply stores the graph family $\{\H_x\text{~| }x\in V\}$.
%for each $s\in V$, a $(\lambda,k)$-FT-BFP with respect to $s$ as the source.
%We denote this subgraph by notation $\H_s$.
%This takes a total of $O(2^k\lambda n^2)$ space.
Given any query vertex-pair $(x,y)$ and any set $F$ of $k$ edges, we 
compute the $(x,y)$-max-flow in $\H_x$ by employing the max-flow algorithm of Chen et al.~\cite{ChenKLPGS22}.
The time to compute the max-flow is $O(|E(\H_x)|^{1+o(1)})$, which is just $O(2^k\lambda n^{1+o(1)})$.
Note that the total space used is bounded by $O(2^k\lambda n^2)$. Therefore, we have the following theorem.
\begin{theorem}
Given any directed graph $\G=(V,E)$ on $n$ vertices, and any positive constants $\lambda,k\geq 1$, we can preprocess $G$ in polynomial time to build an $O(n^2)$ size data structure that, given any query vertex-pair $(x,y)$ and any set $F$ of $k$ edges, can determine the $(x,y)$-$\lambda$-reachability in $\G\setminus F$ in 
$O(n^{1+o(1)})$ time.
\end{theorem}

\subsection{\ftbfp s for graphs with non-unit capacities}

We have shown till now that for any digraph $\G$ with unit capacities,
one can compute a $(\lambda,k)$-\ftbfp ~with $O(2^k\lambda n)$ edges.
We shall now show how to extend this result to a digraph with integer edge capacities
such that flow values up to $\lambda$ are preserved under bounded capacity decrement.

Let us first formalize the notion of $\ftbfp$ under capacity decrement function.

\begin{definition}%[$(\lambda,k)$-FT-BFP]
Let $\G=(V,E,c)$ be a directed flow graph such that capacity %$c(x,y)$
of any edge %$(x,y)$ 
is a positive integer, and 
let $s\in V$ be a designated source vertex. A subgraph $\H=(V,E_0\subseteq E)$ of $\G$
is said to be a $(\lambda,k)$-Fault-Tolerant Bounded-Flow-Preserver if for any 
capacity decrement function $I:E(G)\rightarrow \mathbb{N}$ satisfying $\sum_{e\in E(G)}I(e)\leq k$,
the following holds for the capacity function $c^*$ defined as $c^*(e)=c(e)-I(e)$, for $e\in E$:
\\
For every $t\in V$, 
$$\maxflow(s,t,\H|c^*)
=
\begin{cases} 
\maxflow(s,t,\G|c^*)&\text{if~~}\maxflow(s,t,\G|c^*)\leq \lambda,\\
\text{At least }\lambda,&\text{otherwise};\\
\end{cases}$$
where, $\H|c^*$ and $\G|c^*$ are respectively the graphs $\H$ and $\G$ with capacity function $c^*$.
\end{definition}

\vspace{2mm}

Let us now discuss the construction of $(\lambda,k)$-\ftbfp s. 
Let $\G=(V,E,c)$  be a digraph with integer edge capacities.
We first transform $\G$ into a multigraph $\G^*$ %with unit edge-capacities
by replacing an edge $(x,y)$ of capacity $c(x,y)$
by exactly $c(x,y)$ copies of edge $(x,y)$ of unit-capacity.
Thus, for vertex $v\in V$, the $s$ to $v$ max-flow in graphs $\G$ and $\G^*$
are identical.

Now, let $\H^*$ be a $(\lambda,k)$-\ftbfp ~of multigraph $\G^*$.
Then, a $(\lambda,k)$-\ftbfp ~of $\G$, say $\H=(V,E_0,c)$,
can be obtained by simply retaining all those edges whose
multiplicity in $\H^*$ is non-zero.
The graph $\H$ will indeed be a $(\lambda,k)$-\ftbfp ~of $\G$
since a bounded capacity decrement in $\G$ corresponds 
to $k$-edge failures in $\G^*$.
\vspace{1mm}

% the mandatory LIPIcs bibstyle
\bibliographystyle{plainurl} 
\bibliography{references}

\begin{thebibliography}{10}

\bibitem{AbboudGIKPTUW19}
Amir Abboud, Loukas Georgiadis, Giuseppe~F. Italiano, Robert Krauthgamer, Nikos
  Parotsidis, Ohad Trabelsi, Przemyslaw Uznanski, and Daniel Wolleb{-}Graf.
\newblock Faster algorithms for all-pairs bounded min-cuts.
\newblock In Christel Baier, Ioannis Chatzigiannakis, Paola Flocchini, and
  Stefano Leonardi, editors, {\em 46th International Colloquium on Automata,
  Languages, and Programming, {ICALP} 2019, July 9-12, 2019}, volume 132 of
  {\em LIPIcs}, pages 7:1--7:15, 2019.

\bibitem{BaswanaBP22}
Surender Baswana, Koustav Bhanja, and Abhyuday Pandey.
\newblock Minimum+1 (s, t)-cuts and dual edge sensitivity oracle.
\newblock In Mikolaj Bojanczyk, Emanuela Merelli, and David~P. Woodruff,
  editors, {\em 49th International Colloquium on Automata, Languages, and
  Programming, {ICALP} 2022, July 4-8, 2022, Paris, France}, volume 229 of {\em
  LIPIcs}, pages 15:1--15:20. Schloss Dagstuhl - Leibniz-Zentrum f{\"{u}}r
  Informatik, 2022.

\bibitem{BCCK19}
Surender Baswana, Shreejit~Ray Chaudhury, Keerti Choudhary, and Shahbaz Khan.
\newblock Dynamic {DFS} in undirected graphs: Breaking the o(m) barrier.
\newblock {\em {SIAM} J. Comput.}, 48(4):1335--1363, 2019.

\bibitem{BaswanaCR:16}
Surender Baswana, Keerti Choudhary, and Liam Roditty.
\newblock Fault-tolerant subgraph for single-source reachability: General and
  optimal.
\newblock {\em SIAM Journal on Computing}, 47(1):80--95, 2018.
\newblock \href {https://arxiv.org/abs/https://doi.org/10.1137/16M1087643}
  {\path{arXiv:https://doi.org/10.1137/16M1087643}}, \href
  {https://doi.org/10.1137/16M1087643} {\path{doi:10.1137/16M1087643}}.

\bibitem{BK13}
Surender Baswana and Neelesh Khanna.
\newblock Approximate shortest paths avoiding a failed vertex: Near optimal
  data structures for undirected unweighted graphs.
\newblock {\em Algorithmica}, 66(1):18--50, 2013.

\bibitem{BGLP16}
Davide Bil{\`{o}}, Luciano Gual{\`{a}}, Stefano Leucci, and Guido Proietti.
\newblock Multiple-edge-fault-tolerant approximate shortest-path trees.
\newblock In {\em 33rd Symposium on Theoretical Aspects of Computer Science,
  {STACS} 2016}, pages 18:1--18:14, 2016.

\bibitem{Chechik13}
Shiri Chechik.
\newblock Fault-tolerant compact routing schemes for general graphs.
\newblock {\em Inf. Comput.}, 222:36--44, 2013.

\bibitem{CLPR09}
Shiri Chechik, Michael Langberg, David Peleg, and Liam Roditty.
\newblock Fault-tolerant spanners for general graphs.
\newblock In {\em Proceedings of the 41st Annual {ACM} Symposium on Theory of
  Computing, {STOC} 2009}, pages 435--444, 2009.

\bibitem{CLPR10}
Shiri Chechik, Michael Langberg, David Peleg, and Liam Roditty.
\newblock {\it f}-sensitivity distance oracles and routing schemes.
\newblock In {\em 18th Annual European Symposium on Algorithms - ESA (1)},
  pages 84--96, 2010.

\bibitem{ChenKLPGS22}
Li~Chen, Rasmus Kyng, Yang~P. Liu, Richard Peng, Maximilian~Probst Gutenberg,
  and Sushant Sachdeva.
\newblock Maximum flow and minimum-cost flow in almost-linear time.
\newblock In {\em 63rd {IEEE} Annual Symposium on Foundations of Computer
  Science, {FOCS} 2022, Denver, CO, USA, October 31 - November 3, 2022}, pages
  612--623. {IEEE}, 2022.

\bibitem{Choudhary16}
Keerti Choudhary.
\newblock An optimal dual fault tolerant reachability oracle.
\newblock In {\em 43rd International Colloquium on Automata, Languages, and
  Programming, {ICALP} 2016}, pages 130:1--130:13, 2016.

\bibitem{DTCR08}
Camil Demetrescu, Mikkel Thorup, Rezaul~Alam Chowdhury, and Vijaya
  Ramachandran.
\newblock Oracles for distances avoiding a failed node or link.
\newblock {\em {SIAM} J. Comput.}, 37(5):1299--1318, 2008.

\bibitem{DK11}
Michael Dinitz and Robert Krauthgamer.
\newblock Fault-tolerant spanners: better and simpler.
\newblock In {\em Proceedings of the 30th Annual {ACM} Symposium on Principles
  of Distributed Computing, {PODC} 2011}, pages 169--178, 2011.

\bibitem{DP09}
Ran Duan and Seth Pettie.
\newblock Dual-failure distance and connectivity oracles.
\newblock In {\em Proceedings of the 20th Annual {ACM-SIAM} Symposium on
  Discrete Algorithms, {SODA} 2009}, pages 506--515, 2009.

\bibitem{FF62}
D.~R. Ford and D.~R. Fulkerson.
\newblock {\em Flows in Networks}.
\newblock Princeton University Press, 2010.

\bibitem{GeorgiadisGIPU17}
Loukas Georgiadis, Daniel Graf, Giuseppe~F. Italiano, Nikos Parotsidis, and
  Przemyslaw Uznanski.
\newblock All-pairs 2-reachability in o(n{\^{}}w log n) time.
\newblock In Ioannis Chatzigiannakis, Piotr Indyk, Fabian Kuhn, and Anca
  Muscholl, editors, {\em 44th International Colloquium on Automata, Languages,
  and Programming, {ICALP} 2017, July 10-14, 2017, Warsaw, Poland}, volume~80
  of {\em LIPIcs}, pages 74:1--74:14. Schloss Dagstuhl - Leibniz-Zentrum
  f{\"{u}}r Informatik, 2017.
\newblock URL: \url{https://doi.org/10.4230/LIPIcs.ICALP.2017.74}.

\bibitem{lokshtanov2019brief}
Daniel Lokshtanov, Pranabendu Misra, Saket Saurabh, and Meirav Zehavi.
\newblock A brief note on single source fault tolerant reachability, 2019.
\newblock \href {https://arxiv.org/abs/1904.08150} {\path{arXiv:1904.08150}}.

\bibitem{10.1007/978-3-642-32512-0_24}
Dana Moshkovitz.
\newblock The projection games conjecture and the np-hardness of ln
  n-approximating set-cover.
\newblock In Anupam Gupta, Klaus Jansen, Jos{\'e} Rolim, and Rocco Servedio,
  editors, {\em Approximation, Randomization, and Combinatorial Optimization.
  Algorithms and Techniques}, pages 276--287, Berlin, Heidelberg, 2012.
  Springer Berlin Heidelberg.

\bibitem{NagamochiI:92}
Hiroshi Nagamochi and Toshihide Ibaraki.
\newblock A linear-time algorithm for finding a sparse k-connected spanning
  subgraph of a k-connected graph.
\newblock {\em Algorithmica}, 7(5{\&}6):583--596, 1992.

\bibitem{Parter15}
Merav Parter.
\newblock Dual failure resilient {BFS} structure.
\newblock In {\em Proceedings of the 2015 {ACM} Symposium on Principles of
  Distributed Computing, {PODC} 2015}, pages 481--490, 2015.

\bibitem{PP13}
Merav Parter and David Peleg.
\newblock Sparse fault-tolerant {BFS} trees.
\newblock In {\em Algorithms - {ESA} 2013 - 21st Annual European Symposium,
  Proceedings}, pages 779--790, 2013.

\bibitem{PP14}
Merav Parter and David Peleg.
\newblock Fault tolerant approximate {BFS} structures.
\newblock In {\em Proceedings of the Twenty-Fifth Annual {ACM-SIAM} Symposium
  on Discrete Algorithms, {SODA} 2014}, pages 1073--1092, 2014.

\bibitem{BrandS19}
Jan van~den Brand and Thatchaphol Saranurak.
\newblock Sensitive distance and reachability oracles for large batch updates.
\newblock In {\em 60th {IEEE} Annual Symposium on Foundations of Computer
  Science, {FOCS} 2019, Baltimore, Maryland, USA, November 9-12, 2019}, pages
  424--435, 2019.

\end{thebibliography}

\end{document}